\documentclass[preprint]{elsarticle}
\usepackage[T1]{fontenc}
\usepackage[utf8]{inputenc}
\usepackage[english]{babel}
\usepackage{amsmath,enumerate}
\usepackage{amssymb}
\usepackage{amscd}
\usepackage{vmargin}
\usepackage{url}
\usepackage{stmaryrd}
\usepackage{tikz,pgf}
\usepackage{subfig}

\usetikzlibrary{decorations.pathmorphing}
\tikzset{snake it/.style={decorate, decoration=snake}}

\newcommand\blank[1][.6em]{%
  \mbox{\kern.06em\vrule height.5ex}%
  \vbox{\hrule width#1}%
  \hbox{\vrule height.5ex}}

\begin{document}
\begin{frontmatter}
\title{Strong edge-coloring of $(3, \Delta)$-bipartite graphs\tnoteref{t1}}
\tnotetext[t1]{This research is partially supported by ANR Grant \textsc{STINT} - ANR-13-BS02-0007.}

\author[LIP]{Julien Bensmail}
\author[LIP]{Aurélie Lagoutte}
\author[LIF]{Petru Valicov}
\address[LIP]{LIP, UMR 5668 ENS Lyon, CNRS, UCBL, INRIA, Universit\'e de Lyon, France}
\address[LIF]{Aix-Marseille Université, CNRS, LIF UMR 7279, 13288, Marseille, France}

\begin{abstract}
A strong edge-coloring of a graph $G$ is an assignment of colors to edges such that every color class induces a matching. We here focus on bipartite graphs whose one part is of maximum degree at most~$3$ and the other part is of maximum degree $\Delta$. For every such graph, we prove that a strong $4\Delta$-edge-coloring can always be obtained.
Together with a result of Steger and Yu, this result confirms a conjecture of Faudree, Gy\'arf\'as, Schelp and Tuza for this class of graphs.
\end{abstract}

 
\end{frontmatter}

\newenvironment{proof}{\par \noindent \textbf{Proof} \\}{\hfill$\Box$}

\newtheorem{theorem}{Theorem}
\newtheorem{lemma}{Lemma}
\newtheorem{conjecture}{Conjecture}
\newtheorem{observation}{Observation}
\newtheorem{claim}{Claim}

\section{Introduction}

One common notion of graph theory is the one of \textit{proper edge-coloring}, which is, given an undirected simple graph $G = (V, E)$, an assignment of colors to the edges such that no two adjacent edges receive the same color.  
A proper edge-coloring can equivalently be seen as a partition of the edges into \textit{matchings}. One can easily convince himself that these matchings are generally not induced. If we want each matching of the partition to be induced, then in every part all edges must be sufficiently far apart in the graph. In this perspective, Fouquet and Jolivet introduced the following stronger notion~\cite{FJ83}: a \textit{strong edge-coloring} of $G$ is a proper edge-coloring such that every two edges joined by another edge are colored differently. Clearly, every color class of a given strong edge-coloring is an induced matching. The least number of colors in a strong edge-coloring is referred to as the \textit{strong chromatic index}, denoted $\chi'_s(G)$ for $G$.

We denote by $\Delta(G)$ (or simply $\Delta$ when no ambiguity is possible) the \textit{maximum degree} of $G$.
If $S$ is a subset of vertices of a graph, we refer to $\Delta(S)$ as the maximum degree of the vertices of $S$. Greedy coloring arguments show that $2\Delta^2 - 2\Delta + 1$ is a naive upper bound on the strong chromatic index of any graph. But so many colors are generally not necessary to obtain a strong edge-coloring. Actually, the tightest upper bound on $\chi'_s(G)$ involving $\Delta$ is believed to be the following.

\begin{conjecture}[Erd\H{o}s and Ne\v{s}et\v{r}il~\cite{EN89}] \label{conjecture:global-strong}
For every graph $G$, we have 

\begin{equation*}
  \chi'_s(G) \leq 
     \begin{cases}
        \frac{5}{4}\Delta^2 & \text{if $\Delta$ is even,} \\
        \frac{1}{4}(5\Delta^2 - 2\Delta + 1)  & \text{if $\Delta$ is odd,}
     \end{cases}
\end{equation*}
which, if true, would be tight as the graphs described on Figure~\ref{fig:erdos_nesetril} achieve these bounds.
\end{conjecture}

\newcommand{\join}{\bowtie} 
\begin{figure}[!t]
\centering
\begin{tikzpicture}[join=bevel,inner sep=0.5mm,scale=0.9]
\node[draw, circle, line width=1pt, minimum width=28pt](45) at (18:2) {$I_1$};
\node[draw, circle, line width=1pt, minimum width=28pt](15) at (90:2) {$I_2$};
\node[draw, circle, line width=1pt, minimum width=28pt](12) at (162:2) {$I_3$};
\node[draw, circle, line width=1pt, minimum width=28pt](23) at (234:2) {$I_4$};
\node[draw, circle, line width=1pt, minimum width=28pt](34) at (306:2) {$I_5$};
\draw[-,dashed,line width=0.8pt] (12) -- (15) -- (45) -- (34) -- (23) -- (12);

 \node[rotate=145,fill=white] (join) at (54:1.62) {\Large $\join$};
 \node[rotate=-145,fill=white] (join) at (126:1.62) {\Large $\join$};
 \node[rotate=108,fill=white] (join) at (198:1.62) {\Large $\join$};
 \node[fill=white] (join) at (270:1.62) {\Large $\join$};
 \node[rotate=72,fill=white] (join) at (342:1.62) {\Large $\join$};
 
 \node[right] at (4,2) {$\bullet$ Every $I_j$ is an independent set.};
  \node[right] at (4,1) {$\bullet$ ‘‘$I_j \join I_{j'}$'' means that $I_j$ is complete to $I_{j'}$.};
 \node[right] at (4,0) {$\bullet$ If $\Delta = 2k$, then $ \left|{I_j}\right|=k$.};
 \node[right] at (4,-1) {$\bullet$ If $\Delta = 2k+1$, then $\left|{I_1}\right|=\left|{I_2}\right|=\left|{I_3}\right|=k$};
 \node[right] at (6.8,-1.6) {  ~and $\left|{I_4}\right|=\left|{I_5}\right|=k+1$.};
\end{tikzpicture}
\caption{Erd\H{o}s and Ne\v{s}et\v{r}il's construction.}
\label{fig:erdos_nesetril}
\end{figure}

\noindent This conjecture was verified for graphs of maximum degree at most~3~\cite{A92,HQT93}, and also considered in other situations~\cite{MR97,C06}. But it remains still widely open in general.

\medskip

In this paper we focus on strong edge-coloring of \textit{bipartite graphs}, which are graphs whose vertex set admits a bipartition into two independent sets. In this context, Conjecture~\ref{conjecture:global-strong} was strengthened to the following by Faudree, Gy\'arf\'as, Schelp and Tuza:

\begin{conjecture}[Faudree \emph{et al.}~\cite{FGST90}]
\label{conj:faudree_bipartite}
For every bipartite graph $G$, we have $\chi'_s(G)\leq \Delta^2$.
\end{conjecture}

Brualdi and Quinn Massey introduced a new notion of edge-coloring -- the incidence coloring of graphs~\cite{BQM93}. They showed a connection of this notion with the one of strong edge-coloring, which made them refine Conjecture~\ref{conj:faudree_bipartite}. 

\begin{conjecture}[Brualdi and Quinn Massey~\cite{BQM93}]
\label{conj:brualdi_bipartite}
For every bipartite graph $G$ with bipartition $A$ and $B$, we have $\chi'_s(G)\leq \Delta(A)\Delta(B)$.
\end{conjecture}

In the spirit of this conjecture, we define a \emph{$(d_A, d_B)$-bipartite graph} to be a bipartite graph with parts $A$ and $B$ such that $\Delta(A)\leq d_A$ and $\Delta(B)\leq d_B$.
Conjectures~\ref{conj:faudree_bipartite} and~\ref{conj:brualdi_bipartite} are still widely open, the second being proved to hold in two specific non-trivial situations. It is first known to hold whenever $G$ is subcubic bipartite:

\begin{theorem}[\label{thm:subcubic}Steger and Yu~\cite{SY93}]
For every $(3, 3)$-bipartite graph $G$, we have $\chi'_s(G)\leq 9$.
\end{theorem}

Later on, Nakprasit solved the case where one part of the bipartition is of small maximum degree, namely at most 2.

\begin{theorem}[\label{thm:2Delta}Nakprasit~\cite{N08}]
For every $(2, \Delta)$-bipartite graph $G$, we have $\chi'_s(G)\leq 2\Delta$.
\end{theorem}

Theorems~\ref{thm:subcubic} and~\ref{thm:2Delta} were proved using a similar proof scheme, first used in~\cite{SY93}. Reusing this idea, we prove the following which, together with the aforementioned previous results, settles a special case of Conjecture~\ref{conj:faudree_bipartite}.

\begin{theorem}
\label{thm:main}
For every $(3, \Delta)$-bipartite graph $G$, we have $\chi'_s(G)\leq 4\Delta$.
\end{theorem}

\section{Proof of Theorem~\ref{thm:main}} \label{section:main}

Let $G$ be a $(3, \Delta)$-bipartite graph with bipartition $A$ and $B$ such that $\Delta(A)\leq 3$. We set $n_B=|V(B)|$. It is sufficient to prove the result for the case where all vertices of $A$ are of degree exactly 3, so let us make this assumption.

We describe $G$ by a (non-unique) $(n_B\times \Delta)$-matrix constructed in the following way:

\begin{itemize}
	\item the rows are indexed by the vertices of $B$ and the columns are indexed by $1, 2, \ldots , \Delta$; 
	
	\item every row with index $b\in B$ contains exactly once every edge incident to $b$ (some cells will be empty if $b$ is of degree strictly less than $\Delta$).
\end{itemize}

\noindent We give an example of a bipartite graph and two such associated matrices in Figure~\ref{fig:ex_bip_matrix}. Note that the order of the edges (and the empty cells, if any) in any row of a matrix can be arbitrary, and we will explain later how to take advantage of it. Assuming an edge $e$ of $G$ lies in cell $(i, j)$ of a matrix, we often refer to the index $j$ as the ‘‘column of $e$'' (with respect to this matrix). 

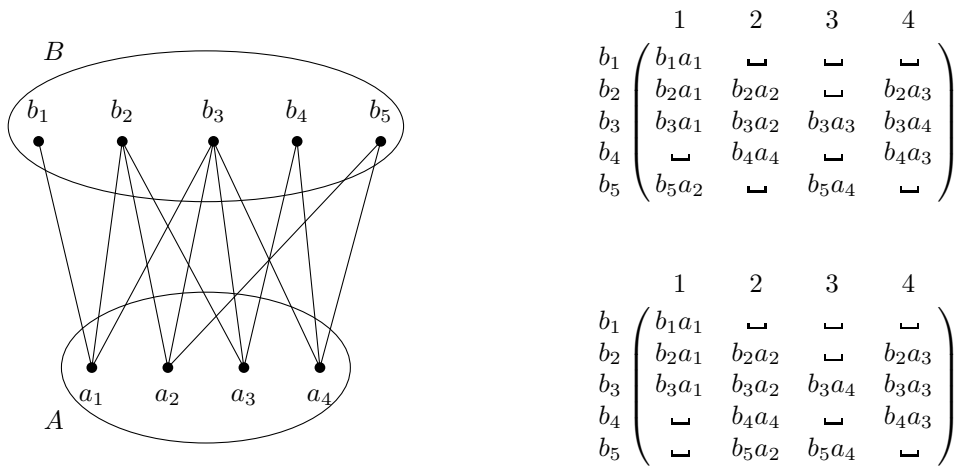
\begin{figure}[!t]
\begin{center}
\begin{tikzpicture}[inner sep=0.5mm]
  \draw[-] (0.5,3.2) ellipse (2.6 and 1);
  \node at (-1.5,4.2) {$B$};
  \node at (-1.5,-0.7) {$A$};
  \draw[-] (0.5,0) ellipse (1.9 and 1);
\draw (0,0) -- (-0.6,3);
\draw (0,0) -- (0.6,3);
\draw (0,0) -- (2.8,3);
\draw (1,0) -- (0.6,3);
\draw (1,0) -- (1.7,3);
\draw (1,0) -- (-0.6,3);
\draw (2,0) -- (0.6,3);
\draw (2,0) -- (1.7,3);
\draw (2,0) -- (2.8,3);
\draw (-1,0) -- (0.6,3);
\draw (-1,0) -- (-1.7,3);
\draw (-1,0) -- (-0.6,3);

\foreach \x / \i in {-1.7/1,-0.6/2,0.6/3,1.7/4,2.8/5}{
  \node[fill,circle] at (\x,3) {};
  \node at (\x,3.4) {$b_{\i}$};
}
\foreach \x / \i in {-1/1,0/2,1/3,2/4}{
  \node[fill,circle] at (\x,0) {};
  \node at (\x,-0.4) {$a_{\i}$};
}
\node at (8,3.5) {
$\bordermatrix{
& 1 & 2 & 3 & 4 \cr
b_1 & b_1a_1 & \blank & \blank & \blank  \cr
b_2 & b_2a_1 & b_2a_2 & \blank & b_2a_3  \cr
b_3 & b_3a_1 & b_3a_2 & b_3a_3 & b_3a_4 \cr
b_4 & \blank & b_4a_4 & \blank & b_4a_3 \cr
b_5 & b_5a_2 & \blank & b_5a_4 & \blank  
}$
};

\node at (8,0) {
$\bordermatrix{
& 1 & 2 & 3 & 4 \cr
b_1 & b_1a_1 & \blank & \blank & \blank  \cr
b_2 & b_2a_1 & b_2a_2 & \blank & b_2a_3  \cr
b_3 & b_3a_1 & b_3a_2 & b_3a_4 & b_3a_3 \cr
b_4 & \blank & b_4a_4 & \blank & b_4a_3 \cr
b_5 & \blank & b_5a_2 & b_5a_4 & \blank  
}$
};
\end{tikzpicture}
\caption{Example of $(3, \Delta=4)$-bipartite graph and two associated matrices. A ‘‘$\blank$'' indicates that the content of the cell is empty.}
\label{fig:ex_bip_matrix}
\end{center}
\end{figure}

Every matrix describing $G$ yields a classification of the vertices of $A$ into three types:

\medskip

 \textbf{Type~1:} vertices whose all incident edges are in the same column,

\medskip

\textbf{Type~2:} vertices whose only two incident edges are in the same column,

\medskip
 
\textbf{Type~3:} vertices whose all incident edges are in different columns.

\medskip

\noindent Since every Type~1 vertex $v$ has all of its three incident edges of the same column, say $i$, calling $i$ the ‘‘column of $v$'' directly makes sense. When considering a Type~2 vertex, we say that its two incident edges located in the same column are \textit{paired}. Its third incident edge is called \textit{lonely}.

\medskip

Since the order of the edges and the empty cells in a given row is arbitrary, different matrices can describe $G$. However, some of them will be better for us, so let us define an order on the matrices and, from now on, consider a \emph{maximum matrix} $M_G$ of $G$. The order is defined as the lexicographical order on $(T_1, T_2)$, where $T_i$ ($i=1,2$) is the number of Type $i$ vertices.
As an illustration of this order, note that, with the first (top) matrix of Figure~\ref{fig:ex_bip_matrix}, only $a_1$ is Type~1, the vertices $a_2$ and $a_3$ are Type~2, while only $a_4$ is Type~3. But this matrix is not maximum in our order as the second (bottom) matrix of Figure~\ref{fig:ex_bip_matrix} describes the same graph but yields three Type~1 vertices ($a_1$, $a_2$ and $a_3$), one Type~2 vertex ($a_4$), and no Type~3 vertex. Thus this second matrix is actually greater in the order which we defined (also note that this matrix is not maximum neither as several other permutations of entries are possible in order to obtain more Type~1 vertices).

Now we give some observations on $M_G$ which will be useful for the coloring process. Most of these observations are straightforward and can be proved by just showing that if some particular situation occurs, then we can perform switches (\emph{i.e.} exchange two edges in a same row) in $M_G$ to get a matrix contradicting the maximality of $M_G$. We provide the proof of Observation~\ref{obs:2Type3_common_neighbour} as an illustration of this statement.

\begin{observation}
\label{obs:vertexB_distinct_colors}
For every $i\in \{1, \ldots, \Delta\}$, every vertex of $B$ has at most one incident edge in column $i$.
\end{observation}

Let $e$ and $e'$ be two edges of $G$. We say that $e$ is \emph{visible} from $e'$ (or $e'$ \textit{sees} $e$) if $e$ and $e'$ are adjacent or share a common adjacent edge. So equivalently a strong edge-coloring is an assignment of colors such that every two edges which are mutually visible are assigned different colors.

\begin{observation}
\label{obs:2vertices_common_neighbour}
If two vertices $a_0$ and $a_1$ of $A$ have no common neighbor in $B$, then every edge incident to $a_0$ sees no edge incident to $a_1$.
\end{observation}

\begin{observation}
\label{obs:2Type3_common_neighbour}
Let $a_0$ and $a_1$ be two Type~3 vertices with incident edges $a_0b_0,a_0b_1,a_0b_2$ and $a_1b_3,a_1b_4,a_1b_2$, respectively. Note that $b_2$ is a common neighbor of $a_0$ and $a_1$. Let $i,j,k$ (respectively $i',j',k'$) be the columns of $a_0b_0,a_0b_1,a_0b_2$ (respectively $a_1b_3,a_1b_4,a_1b_2$). Then $k\notin\{i',j',k'\}$ and $k'\notin\{i,j,k\}$.
\end{observation}

\begin{proof}
Assume by contradiction that one of the situations described in the statement occurs, \emph{e.g.} that without loss of generality we have $k=i'$. Then $M_G$ looks like the first (left) matrix depicted in Figure~\ref{fig:ex_proof}. But then, by switching $b_2a_0$ and $b_2a_1$ in the row indexed by $b_2$, we get the second (right) matrix $M_G'$ depicted in Figure~\ref{fig:ex_proof} which yields the same number of Type~$1$ vertices, but one extra Type~$2$ vertex $a_1$. Therefore, $M_G$ is not maximum -- a contradiction.
\end{proof}

\begin{figure}[!t]
\begin{center}
\begin{tikzpicture}[inner sep=0.5mm]
\node at (0,0) {
$\bordermatrix{
& i & j & k=i' & j' & k' \cr
b_0 & b_0a_0 & - & - & - & -  \cr
b_1 & - & b_1a_0 & - & - & -  \cr
b_2 & - & - & b_2a_0 & - & b_2a_1 \cr
b_3 & - & - & b_3a_1 & - & - \cr
b_4 & - & - & - & b_4a_1 & -
}$
};

\node at (0.35, -1.75) {
$M_G$
};

\node at (8,0) {
$\bordermatrix{
& i & j & k=i' & j' & k' \cr
b_0 & b_0a_0 & - & - & - & -  \cr
b_1 & - & b_1a_0 & - & - & -  \cr
b_2 & - & - & \mathbf{b_2a_1} & - & \mathbf{b_2a_0} \cr
b_3 & - & - & b_3a_1 & - & - \cr
b_4 & - & - & - & b_4a_1 & -
}$
};

\node at (8.35, -1.75) {
$M_G'$
};

\end{tikzpicture}
\caption{An illustration of the proof of Observation~\ref{obs:2Type3_common_neighbour}. The first (left) matrix $M_G$ cannot be maximum, since one can switch two edges on a same row to get the second (right) matrix $M_G'$, which has more Type 2 vertices and the same number of Type 1 vertices. A ‘‘$-$'' indicates that the content of the cell can be arbitrary (either filled or empty).}
\label{fig:ex_proof}
\end{center}
\end{figure}

\begin{observation}
\label{obs:Type3_adjacent_Type2}
Let $a_0$ be a Type~3 vertex with incident edges $e_1,e_2,e_3$ in columns $i,j,k$, respectively. Let $a_1$ be a Type~2 vertex with incident edges $e_4,e_5,e_6$ in columns $i',i',j'$, respectively. If $e_6$ is adjacent to $e_3$, then $j'\notin \{i,j,k\}$ and $i'\neq k$.
\end{observation}

\begin{observation}
\label{obs:Type2_adjacent_Type2}
Let $a_0$ and $a_1$ be two Type~2 vertices. Let $e_1$, $e_2$, $e_3$ (respectively $e_4$, $e_5$, $e_6$) be their incident edges in columns $i$, $i$, $j$ (respectively $i'$, $i'$, $j'$). If $e_3$ is adjacent to $e_4$ or $e_5$, then $i\neq i'$.
\end{observation}

\begin{observation}
\label{obs:Type2_adjacent_Type22}
Let $a_0$ and $a_1$ be two Type~2 vertices. Let $e_1$, $e_2$, $e_3$ (respectively $e_4$, $e_5$, $e_6$) be their incident edges in columns $i$, $i$, $j$ (respectively $i'$, $i'$, $j'$). If $e_3$ is adjacent to $e_6$, then $j \neq i'$.
\end{observation}

\begin{observation}
\label{obs:2Type1_adjacent}
Let $a_0$ and $a_1$ be two Type~1 vertices of columns $i$ and $j$, respectively. If $a_0$ and $a_1$ have a common neighbor, then $i\neq j$.
\end{observation}

\begin{observation}
\label{obs:a_branch}
Let $a$ be a Type~2 vertex with incident edges $e_1$, $e_2$, $e_3$, where $e_1$ is the lonely edge in column $j$. Then at least one of $e_2$ or $e_3$ is not adjacent to a lonely edge of column $j$ different from $e_1$.
\end{observation}

We now describe the coloring process which will yield a strong $4\Delta$-edge-coloring $c$ of $G$. Each edge $e$ will be given a color $c(e)=(i,j)$, where $j\in\{1,\ldots,\Delta\}$ is fixed as the column of $e$ in $M_G$ and $i\in\{1,2,3,\ast\}$ is to be set in the coloring process.
So, in what follows, by ‘‘coloring an edge'' we mean assigning a value to $i$. 

The coloring process mainly consists in coloring the edges of $G$ successively without creating any conflict, \emph{i.e.} in such a way that every resulting partial edge-coloring remains strong. Its successive steps are the following:

\medskip

\noindent \textbf{\emph{Coloring Procedure:}} \\
\indent \textbf{Step~1:} color the edges incident to Type~1 vertices. \\
\indent \textbf{Step~2:} color the paired edges incident to Type~2 vertices. \\
\indent \textbf{Step~3:} color the edges incident to Type~3 vertices. \\
\indent \textbf{Step~4:} color the lonely edges incident to Type~2 vertices.

\medskip

\noindent In order to show that this coloring procedure is almost optimal somehow, we will impose ourselves the constraint that the ‘‘special'' color $\ast$ is used during Step~4 only. This will show that $3\Delta$ colors are sufficient to color all edges considered during Steps~1 to~3. 

The first three steps will be performed greedily, while the last one requires a careful analysis of the structure of the remaining non-colored edges. The rest of this section is dedicated to explanations on why this procedure can be achieved correctly, \emph{i.e.} why there is always an available color for an edge considered at any of the four steps.


\subsection*{\textbf{Step 1:} color the edges incident to Type~1 vertices.}

For each Type~$1$ vertex with incident edges $e_1$, $e_2$ and $e_3$, just color $e_1$, $e_2$ and $e_3$ greedily (\emph{i.e.} properly) with $\{1,2,3\}$. The obtained edge-coloring is also strong, which follows directly from Observations~\ref{obs:2vertices_common_neighbour} and~\ref{obs:2Type1_adjacent}.


\subsection*{\textbf{Step 2:} color the paired edges incident to Type~2 vertices.}

Once again, for each Type~2 vertex with incident paired edges $e_1$ and $e_2$, we just color $e_1$ and $e_2$ greedily, in such a way that no conflict arises with the already colored edges. The following lemma shows that this is always possible, \emph{i.e.} that, after Step~$1$ and at any moment of Step~$2$, there is always (at least) one color available for any considered paired edge.

\begin{lemma}
\label{lemma:step2}
After performing Step 1 and any number of iterations of Step~2, for each Type~2 vertex which was not considered yet, there are always at least two colors available among $\{1,2,3\}$ for each of its paired edges.
\end{lemma}

\begin{proof}
Let $a$ be a Type~2 vertex, $ab_1$ and $ab_2$ be its paired edges situated in column, say, $j$ of $M_G$, and $ab_0$ be the lonely edge. 
Let us count the number of already colored edges in column $j$ visible from $ab_1$ or $ab_2$. We prove that there is at most one such edge, which moreover is incident to $b_0$.
First recall that, according to Observation~\ref{obs:vertexB_distinct_colors}, none of the edges incident to $b_1$ or $b_2$, except $ab_1$ and $ab_2$, are in column $j$. 
 Consider the neighbors of $b_1$ and $b_2$ distinct from $a$. Without loss of generality we consider one of them, say $a_1$ - neighbor of $b_1$, and assume for contradiction that $a_1$ has at least one already colored incident edge $e$ in column $j$. As mentioned previously, according to Observation~\ref{obs:vertexB_distinct_colors}, $e$ cannot be $a_1b_1$. Thus $a_1$ cannot be a Type~1 vertex. Moreover, if $a_1$ is a Type~3 vertex, then $e$ has not been colored yet. The same happens if $a_1$ is a Type~2 vertex and $e$ is a lonely edge. The last case occurs when $a_1$ is a Type~2 vertex and $e$ is a paired edge: by Observation~\ref{obs:vertexB_distinct_colors}, edge $a_1b_1$ has to be lonely, and then Observation~\ref{obs:Type2_adjacent_Type2} yields a contradiction.
 Now observe that $b_0$ has at most one incident edge in column $j$ by Observation~\ref{obs:vertexB_distinct_colors}. Consequently, at any moment while performing Step 2 of the procedure, two colors among $\{1, 2, 3\}$ are available for $ab_1$ and $ab_2$.
\end{proof}


\subsection*{\textbf{Step 3:} color the edges incident to Type~3 vertices.}

Once again, a correct extension of the partial strong edge-coloring to the edges incident to the Type~3 vertices can be obtained greedily. The following lemma shows that available colors exist for any edge considered during the procedure.

\begin{lemma}
\label{lemma:step3}
After performing Step 2 and any given number of iterations of Step 3, for each edge incident to any given Type~3 vertex there is at least one available color among $\{1,2,3\}$.
\end{lemma}

\begin{proof}
Let $a$ be a Type~3 vertex with neighbors $b_0$, $b_1$ and $b_2$, and let $j$ be the column of $ab_0$.
Let us count the number of edges visible from $ab_0$, which are already colored and in column $j$. We prove that there can be at most two of them.
Due to Observation \ref{obs:vertexB_distinct_colors}, vertices $b_1$ and $b_2$ can each have at most one incident edge in column $j$. Let $a_0$ be a neighbor of $b_0$ and suppose for contradiction that $a_0$ has an incident edge $e$ in column $j$. By Observation~\ref{obs:vertexB_distinct_colors}, edge $a_0b_0$ cannot be in column $j$, and thus $a_0$ is not a Type~1 vertex. For the same reason, if $a_0$ is of Type~2, edge $a_0b_0$ cannot be paired with $e$; moreover, by Observation~\ref{obs:Type3_adjacent_Type2}, edge $a_0b_0$ cannot be a lonely edge, so $e$ is the lonely edge of $a_0$, and thus is not colored yet. Finally $a_0$ cannot be a Type~3 vertex according to Observation~\ref{obs:2Type3_common_neighbour}. Thus at least one color among $\{1,2,3\}$ is available for $a_0b_0$.
\end{proof}


\subsection*{\textbf{Step 4:} color the lonely edges incident to Type~2 vertices.}

Before explaining how to color the lonely edges explicitly, we first introduce some notions and raise some observations about how these edges appear in $G$.

Let $F$ be a subset of edges of $G$. The \emph{subgraph induced by $F$} is the subgraph induced by the endpoints of the edges of $F$. For each column $j$ of $M_G$, we define the \emph{component of $j$}, denoted $\mathcal{C}_j$, as the subgraph of $G$ induced by the set of lonely edges of column $j$.
Since $G$ is bipartite, observe that every cycle of $\mathcal{C}_j$ have even length. We call a cycle $v_0v_1v_2\ldots v_{k-1}v_0$ of $\mathcal{C}_j$ \emph{alternate} if exactly half of its edges are lonely in column $j$ and, for every pair of consecutive edges $v_iv_{i+1}$ and $v_{i+1}v_{i+2}$ (where $i$ is taken modulo $k$), one is lonely in column $j$ and the other is not (\emph{i.e.} the lonely edges of column $j$ on the cycle are non-adjacent). 
Similarly, we say that a path of $\mathcal{C}_j$ is \textit{alternate} if for every pair of adjacent edges of the path, one of them is lonely in column $j$ and the other is not.
We prove below that each $\mathcal{C}_j$ has a very specific structure. We first start with a direct consequence of Observation \ref{obs:a_branch}.

\begin{observation}
\label{obs: vertex in A with both paired edges in Cj}
Let $j$ be a column of $M_G$. Every Type~2 vertex $a\in A$ appearing in $\mathcal{C}_j$ cannot have both its paired edges in $\mathcal{C}_j$.
\end{observation}


\begin{lemma}
\label{lemma:component_cycles}
Let $j$ be a column of $M_G$. Every connected component of $\mathcal{C}_j$ has at most one cycle. Moreover, if this cycle exists, then it must be alternate.
\end{lemma}

\begin{proof}
Observe that $\mathcal{C}_j$ has no lonely edge $e$ which is not in column $j$. Otherwise, since we are considering the component $\mathcal{C}_j$, the endpoint of $e$ in part $A$ would be incident to a lonely edge of column $j$ contradicting the definition of a Type~2 vertex. Therefore, from now on in this proof, when speaking about lonely edges of $\mathcal{C}_j$ we will refer to lonely edges of column $j$.

First we show that all cycles in $\mathcal{C}_j$ are alternate. 
Suppose by contradiction that there is a cycle in $\mathcal{C}_j$ which is non-alternate. Observe first that there cannot be two adjacent lonely edges in $\mathcal{C}_j$ (otherwise there would be two lonely edges incident to a same vertex in $A$ or $B$, which is impossible by the definition of a Type~2 vertex and Observation~\ref{obs:vertexB_distinct_colors}). Thus by hypothesis the non-alternate cycle must have two non-lonely adjacent edges $e$ and $e'$ sharing a same vertex $v$.

Observe first that $v$ cannot be in part $A$: otherwise, these two non-lonely edges $e$ and $e'$ would be the paired edges of $v$, a contradiction with Observation \ref{obs: vertex in A with both paired edges in Cj}.
Now suppose that $v$ is in part $B$. We denote the non-alternate cycle by $C=b_0a_0b_1a_1...b_ka_kb_0$, where each vertex $a_i$ (resp. $b_i$) belongs to $A$ (resp. $B$). Assume $v=b_0$, as well as $e=b_0a_0$ and $e'=b_0a_k$. Then, since no vertex $a_i$ has its two paired edges along $C$ (according to Observation \ref{obs: vertex in A with both paired edges in Cj}), we get that $a_0b_1$ is lonely. Now, since two lonely edges cannot be adjacent, $b_1a_1$ is not lonely. Repeating these arguments along the edges of $C$, we get that every edge $b_ia_{i}$ with $0\leq i \leq k$ is non-lonely, while every $a_ib_{i+1}$ for $0\leq i \leq k-1$ is lonely. Then we get that the two edges incident to $a_k$ along $C$ are not lonely, which contradicts Observation \ref{obs: vertex in A with both paired edges in Cj}.

Therefore, all the cycles of the component are alternate.

\medskip

Now we prove that there can be only one alternate cycle (if any) in every connected component of $\mathcal{C}_j$. Suppose by contradiction that there are two alternate cycles in a connected component of $\mathcal{C}_j$. We show the following properties about these two cycles to end up with a contradiction:
\begin{enumerate}
\item the two cycles cannot share a vertex without sharing an edge,
\item the two cycles cannot share an edge,
\item the two cycles cannot be joined by a path in the component.
\end{enumerate}

The first property follows from the fact that the two cycles are alternate and there cannot be two adjacent lonely edges in a same component. Suppose by contradiction that the second property is false, \emph{i.e.} that two cycles share an edge. Let $P=v_1\cdots v_k$ be one longest alternate path shared by theses cycles. Observe that $v_1v_2$ must be lonely (since otherwise there would be two adjacent lonely edges) and $v_1$ must have two other incident non-lonely edges - one in each of the two cycles. We call these edges $e$ and $e'$ respectively, and observe then that $v_1\notin A$ thanks to Observation \ref{obs: vertex in A with both paired edges in Cj}.
However, by the same arguments, $v_k$ must be in $B$ as well, and $v_{k-1}v_k$ must be a lonely edge. Then $P$ is an alternate path of odd length between $v_1$ and $v_k$ which are both in part $B$, a contradiction since $\mathcal{C}_j$ is bipartite.

Finally, in order to show the third property, suppose by contradiction that there is a path connecting the two cycles in the component. Consider in particular the shortest path $P$ with extremities $u$ and $v$, where $u$ lies on the first cycle while $v$ lies on the second one. Recall that the cycles are alternate, and thus one edge of the first cycle incident to $u$ is lonely, and similarly for $v$ with respect to the second cycle. Recall also that $P$ is a subgraph of $\mathcal{C}_j$. Therefore, by Observation~\ref{obs: vertex in A with both paired edges in Cj}, none of $u$ and $v$ can be in part $A$: otherwise, they would have two paired edges in $\mathcal{C}_j$ - one on the cycle, and one on $P$.

Let us hence denote $P = b_0a_0b_1a_1...b_k$, where $u=b_0$ and $v=b_k$, and $k \geq 2$ is even. Now consider the successive vertices of $P$, \textit{i.e.} from $b_0$ to $b_k$.
By Observation \ref{obs:vertexB_distinct_colors}, the edge $b_0a_0$ cannot be lonely since $b_0$ already has an incident lonely egde on the first cycle. Now, $a_0b_1$ has to be lonely, since otherwise $a_0$ would have both its paired edges in $\mathcal{C}_j$. Repeating the same arguments until we reach $b_k$, we get that every edge $b_ia_{i}$ is not lonely, while every edge $a_ib_{i+1}$ is lonely, for $0\leq i \leq k-1$. Then $b_k$ is incident to two lonely edges (one is $a_{k-1}b_k$ and the other one is on the second cycle), a contradiction with Observation~\ref{obs:vertexB_distinct_colors}. 
\end{proof}
\medskip

We now explain how to color the lonely edges in order to finish the coloring of $G$. Recall that during this step, we allow the use of the special color $\ast$. Consider every successive value of $j \in \{1, ..., \Delta\}$. We may assume that $\mathcal{C}_j$ is connected (if not, apply the procedure below component-wisely). The lonely edges of $\mathcal{C}_j$ are colored in up to two phases as follows:

\medskip

\textbf{Phase~1:} In case $\mathcal{C}_j$ has an induced cycle $C$, it is unique and alternate according to Lemma~\ref{lemma:component_cycles}. Let $C=a_1b_1a_2b_2...a_{k}b_{k}a_1$ be this cycle, where $a_1b_1, a_2b_2,\ldots,a_{k}b_{k}$ are its lonely edges and $a_i\in A$ (resp. $b_i\in B$) for $1\leq i \leq k$. Then greedily color the edges $a_1b_1,\ldots,a_{k-1}b_{k-1}$, in this order, with colors among $\{1, 2, 3\}$ in order to obtain a partial strong edge-coloring. Color the remaining lonely edge $a_{k}b_{k}$ with color $\ast$. 

\medskip

\textbf{Phase~2:} If $C$ exists, then first remove its edges to get a (possibly empty) forest. Each tree $T$ of the forest will have as a root a vertex $r$ of $C$. If the component had no cycle $C$, we designate an arbitrary node of $T$ to be the root $r$.
Then greedily color with colors among $\{1, 2, 3, *\}$ the remaining uncolored lonely edges of $T$ as they are encountered during a Breadth-First Search (BFS) algorithm performed from $r$.

\medskip

The following two results show that Phases~1 and~2 can always be performed correctly.

\begin{lemma}
During Phase~1, for every lonely edge of $C$ there is at least one available color among $\{1,2,3,\ast\}$.
\end{lemma}
\begin{proof}
Assume that the edges $a_1b_1,\ldots, a_{i-1}b_{i-1}$ have already been colored and let $a_{i}b_{i} \neq a_{k}b_{k}$ be the considered lonely edge. Recall that, due to our ordering, the edge $a_{i+1}b_{i+1}$ is uncolored. Recall also that no other edge adjacent to $a_{i}b_{i}$ in $G$ is in column $j$ of $M_G$ (according to Observation~\ref{obs:vertexB_distinct_colors} and the definition of a Type~2 vertex). 
Let us count the number of edges visible from $a_ib_i$ which are already colored and in column $j$. Let us prove that there can be at most two of them. One of them is $a_{i-1}b_{i-1}$.
Let $b$ be the third neighbor of $a_{i}$. By Observation~\ref{obs:vertexB_distinct_colors}, at most one of the edges incident to $b$ can be in column $j$. Now consider a neighbor $a$ of $b_i$ (different from $a_i$) and assume it has an incident edge $e$ in column $j$. By Observation~\ref{obs:vertexB_distinct_colors}, $a$ cannot be a Type 1 vertex, nor a Type 2 vertex where $e$ would be paired with $ab_i$. By Observation~\ref{obs:Type2_adjacent_Type22}, if $a$ is a Type 2 vertex, then $e$ is lonely and thus not yet colored: indeed, there exists at most one cycle per component (according to Lemma \ref{lemma:component_cycles}), and, for now, we have colored only lonely edges involved in a cycle. Finally $a$ cannot be a Type 3 vertex according to Observation~\ref{obs:Type3_adjacent_Type2}. Therefore, one color among $\{1,2,3\}$ is available to color $a_{i}b_{i+1}$. As for $a_{k}b_{k}$, no other edge of the same connected component of $\mathcal{C}_j$ is colored with~$\ast$, so coloring this edge cannot create any conflict (note that $a_kb_k$ can have three visible edges in column $j$ and thus none of $\{1,2,3\}$ may be available). This completes the proof.
\end{proof}

\begin{lemma}
During the BFS algorithm in Phase~2, for every lonely edge of $T$ there is at least one available color among $\{1,2,3,\ast\}$.
\end{lemma}

\begin{proof}
Consider a Type~2 vertex $a\in A$ with lonely edge $ab_0\in T$, where $a \in A$ and $b_0 \in B$. So $a$ is Type~2 with paired edges $a_0b_1$ and $a_0b_2$. Then $b_1$ and $b_2$ can be each incident to at most one edge in column $j$, and each of these two edges may be colored already (for example, if both $b_1$ and $b_2$ are adjacent to a Type~1 vertex in column $j$).

We now prove that the other edges visible from $ab_0$ and in column $j$ have to be lonely, and that at most one of them is already colored.
Consider any edge $b_0a_0$ different from $ab_0$ and adjacent to an edge $e$ in column $j$. By Observation \ref{obs:vertexB_distinct_colors}, $b_0a_0$ cannot be in column $j$. Then $a_0$ is either of Type~2 or Type~3. Actually, $a_0$ cannot be of Type~3 according to Observation~\ref{obs:Type3_adjacent_Type2}. Also, according to Observation~\ref{obs:Type2_adjacent_Type22}, $a_0$ cannot be of Type~2 with $b_0a_0$ being lonely and its paired edges being of column~$j$. So necessarily $a_0$ is of Type~$2$ with lonely edge $e$ in column~$j$. Observe that $b_0$ may have several neighbors playing the same role as $a_0$, \emph{i.e.} incident to an edge $e'$ in column $j$, but then the same argument applies and $e'$ is lonely.

The important remark to raise is that the BFS algorithm performed on $T$ from $r$ ensures that, whenever a lonely edge $e$ is treated, at most one lonely edge of $T$ (and $C$, if it exists) visible from $e$ has already been colored: indeed, assume first that $e$ is not adjacent to the root and call $v_{\uparrow}$ (resp. $v_{\downarrow}$) the endpoint of $e$ which is closer (resp. further) to the root $r$. Then $e$ is the only lonely edge adjacent to $v_{\uparrow}$; call $v_{\uparrow\uparrow}$ the father of $v_{\uparrow}$ in the tree: $v_{\uparrow\uparrow}$ has only one incident lonely edge in column $j$ (which happens to be colored before $e$); finally, any subtree rooted at a son of $v_{\uparrow}$, or rooted at $v_{\downarrow}$ is not colored yet. So $e$ sees at most one already colored lonely edge.
Let us now deal with a lonely edge $e$ incident to the root $r$: if $\mathcal{C}_j$ had a cycle, then we chose the root $r$ to be on $C$; consequently $r$ already has an incident lonely edge on the cycle, a contradiction with Observation \ref{obs:vertexB_distinct_colors} and the definition of a Type~2 vertex. Otherwise, the component had no cycle, and thus no lonely edge visible from $e$ has already been colored.

\end{proof}


\section{Conclusion and possible improvements}

In this paper, we have proved that, for every $(3, \Delta)$-bipartite graph $G$, we have $\chi'_s(G) \leq 4\Delta$. This result, together with Theorem~\ref{thm:subcubic}, confirms Conjecture~\ref{conj:faudree_bipartite} for this specific family of bipartite graphs. We however believe that our upper bound should not be tight, as stated in Conjecture~\ref{conj:brualdi_bipartite} where $3\Delta$ is conjectured to be the right bound.

\subsection*{Avoiding using $\ast$}

Maybe the  upper bound we have obtained, could be improved by refining the coloring procedure introduced in Section~\ref{section:main}. To do so, one would have to find a way to do without color~$\ast$, \emph{i.e.} color every lonely edge of $\mathcal{C}_j$ with ‘‘regular'' color $1$, $2$ or $3$. One optimistic reason why this should be possible is that each such color is only used when $1$, $2$ and $3$ are all forbidden, that is when coloring the connected components of the $\mathcal{C}_j$'s during Step~4.

On the one hand, color~$(\ast, j)$ is always used, in Phase~1 of Step~4, once for each alternate cycle of $\mathcal{C}_j$. But this use of $(\ast,j)$ is sometimes not necessary. The main purpose for us to systematically use it is to facilitate and lighten the proof of Theorem~\ref{thm:main} by avoiding a tedious case analysis. But one may  note that the only bad situation, that is when the use of color $(\ast, j)$ might be necessary to color the lonely edges of $C$, is when the following three conditions are satisfied:
\begin{itemize}
\item $C$ is of length $2k$ with $k \geq 3$ odd;
\item every vertex $a_i\in A$ of $C$ is at distance 2 from a Type~1 vertex $a'_i$ of column $j$ --	 call $e'_i$ the edge which is not in $C$ and which joins $a'_i$ and the common neighbor of $a_i$ and $a'_i$;
\item and every edge $e'_i$ has been assigned exactly the same color $(i,j)$ in Step~1 with $i\in \{1,2,3\}$.
\end{itemize}

On the other hand, color $(\ast, j)$ may also be used during Phase~2 of Step~4 to color a lonely edge of a tree $T$ of $\mathcal{C}_j - C$. A careful analysis shows that actually color $(\ast, j)$ may only be needed for lonely edges incident to a leaf of $T$, and if the around vertices are colored in an unfavourable way (typically when several Type~1 vertices surround the leaf).

We believe that if it would be possible to decrease the number of colors used in our procedure, these two bad cases above should be the ones to tackle.
To this regard, choosing $M_G$ among all maximal matrices so that it meets additional convenient properties such as minimizing the number of alternate cycles would be interesting to investigate. Also, it is worth pointing out that many tasks of the coloring process are performed arbitrarily (\textit{e.g.} coloring the edges during Steps~$1$ to~$3$, the choice of $r$ during Phase~$2$ of Step~$4$, etc.). Searching for better choices would be another promising perspective.

\subsection*{From $3$ to higher values of $\Delta(A)$}

An interesting perspective of research is to investigate whether the coloring scheme we have used herein may be generalized to larger values of $\Delta(A)$. One could indeed, based on some maximum matrix $M_G$ describing $G$, organize the incident edges of every vertex in $A(G)$ into \textit{maximal groups of paired edges}, \textit{i.e.} edges in a same column of $M_G$, and generalize the coloring scheme described in Section~\ref{section:main}. Namely, one could first color the maximal groups of $\Delta(A)$ paired edges (which correspond to the notion of Type~1 vertex herein), then color the maximal groups of $\Delta(A)-1$ paired edges, and so on, and show that such strong extensions exist according to generalized versions of Observations~\ref{obs:vertexB_distinct_colors} to 9. Following the same idea as in our proof of Theorem~\ref{thm:main}, the algorithm to color the graph would require $\Delta(A)+1$ steps. But the success of this task does not seem immediate to us. In particular, the last step of the new procedure seems hard to define, a simple adaptation of Step~4 from the proof of Theorem~\ref{thm:main} being not clear. This is due to the fact that expressing accurately how the maximal groups of paired edges are organized in $G$ in general, is not easy.

\subsection*{Complexity matters}

For computational complexity, our proof yields a polynomial-time algorithm to deduce a strong $4\Delta$-edge-coloring of a given $(3, \Delta)$-bipartite graph $G$. Indeed, first note that a coloring can easily be obtained once $M_G$ is known, since assigning a color to an edge then just requires to check what are its neighboring colors. We start the coloring process with any matrix $M_G$ (not necessarily maximal). Then during the coloring process, if at some particular step the coloring cannot be achieved, then that would imply that $M_G$ is not maximal (since one of the observations would not be satisfied). Moreover, in this case we would know which are the entries of $M_G$ to be permuted in order to obtain another matrix $M'_G$ which would be greater than $M_G$. Thus we restart the coloring process on $M'_G$. In the worst case, the coloring process will be restarted $O(|V(A)|^2)$ times until we reach a matrix $M'_G$ which is maximal. This clearly shows that the coloring is obtained in polynomial time.

On the other hand, it turns out that obtaining $M_G$ is an \textsf{NP}-hard problem in general. In order to prove this statement, let us introduce the following problem.

\medskip

\noindent \textsc{Maximum Number of Type~1 Vertices} \\
\noindent Instance: a bipartite graph $G$ and an integer $\ell \geq 1$.\\
\noindent Question: does there exist a matrix describing $G$ yielding at least $\ell$ Type~1 vertices?

\medskip

Our statement above follows from a polynomial-time reduction from the following problem, where a \textit{properly $k$-vertex-colorable graph} is a graph admitting a \textit{proper $k$-vertex-coloring}, that is a partition of its vertices into $k$ independent sets (\textit{i.e.} with no adjacent vertices).

\medskip

\noindent \textsc{Maximum Properly $k$-Vertex-Colorable Subgraph} \\
\noindent Instance: a graph $G$ and an integer $\ell \geq 1$. \\
\noindent Question: does there exist a properly $k$-vertex-colorable subgraph of $G$ with at least $\ell$ vertices?

\medskip

\textsc{Maximum Properly $2$-Vertex-Colorable Subgraph} is known to remain \textsf{NP}-complete when its input graph is of maximum degree~$3$ (see~\cite{CNR89}). Using this fact, we prove the following result establishing the hardness of \textsc{Maximum Number of Type~1 Vertices}.
\begin{theorem}
\textsc{Maximum Number of Type~1 Vertices} is \textsf{NP}-complete, even when restricted to $(3, 2)$-bipartite graphs.
\end{theorem}
\begin{proof}
Given a matrix $M_G$ describing a graph $G$ (which, obviously, has size polynomial in the number of vertices of $G$), one can compute in polynomial time the number of Type~$1$ vertices yielded by $M_G$. So \textsc{Maximum Number of Type~1 Vertices} is an \textsf{NP} problem. 

We now prove the \textsf{NP}-hardness of \textsc{Maximum Number of Type~1 Vertices}. Consider an instance of \textsc{Maximum Properly $2$-Vertex-Colorable Subgraph}, \textit{i.e.} a graph $G$ of maximum degree~$3$ together with an integer $\ell \geq 1$. From $G$, we construct a $(3, 2)$-bipartite graph $H$ such that the number of vertices in a maximum properly $k$-vertex-colorable subgraph of $G$ is exactly equal to the number of Type~$1$ vertices yielded by a maximum matrix $M_H$ describing $H$. Hence, $(H,\ell)$ will be a positive instance of \textsc{Maximum Number of Type~1 Vertices} if and only if $(G, \ell)$ is a positive instance of \textsc{Maximum Properly $2$-Vertex-Colorable Subgraph}.
We construct $H$ as the 1-subdivision of $G$, namely $H=A\cup B$, as follows:
\begin{itemize}
\item for every vertex $u$ of $G$, add a vertex $a_u$ to $H$,
\item for every edge $uv$ of $G$, add one vertex $b_{u,v}$ to $H$,
\item $A = \{a_u : u \in V(G)\}$ and $B = \{b_{u,v} : uv \in E(G)\}$,
\item for every edge $uv$ of $G$, add the edges $a_ub_{u,v}$ and $a_vb_{u,v}$ to $H$.
\end{itemize}
\noindent Clearly $\Delta(A) \leq 3$ and $\Delta(B) = 2$, so $H$ is a $(3, 2)$-bipartite graph. Besides, the reduction is achieved in polynomial time since the number of vertices of $H$ is $|V(G)|+|E(G)|$. Note that every two adjacent vertices $u$ and $v$ of $G$ are directly depicted in $H$ by the two vertices $a_u$ and $a_v$ which are at distance exactly~$2$ (because of $b_{u,v}$). So $u$ and $v$ cannot be assigned the same color by a partial proper $2$-vertex-coloring of $G$ while $a_u$ and $a_v$ cannot be Type~$1$ vertices of a same column of $M_H$, and vice-versa. From this fact, assuming color, say, $1$ is liken to column~$1$ of $M_H$, coloring~$1$ a vertex $u$ of $G$ is equivalent to having $a_u$ being a Type~1 vertex of column~$1$ of $M_H$. Because $\Delta(B) = 2$, note that $M_H$ has exactly two columns by definition, and so we can define a straight equivalence between the two colors used to color~$G$ and the two columns of~$M_H$. The equivalence between the two instances then follows.
\end{proof}

\end{document}